\documentclass[aps,prl,twocolumn,superscriptaddress,floatfix,nofootinbib,showpacs,longbibliography]{revtex4-2}

\usepackage[utf8]{inputenc}
\usepackage[T1]{fontenc}     
\usepackage[british]{babel}  
\usepackage[sc,osf]{mathpazo}
\usepackage{times}
\usepackage[table]{xcolor}
\usepackage[scaled=0.86]{berasans}  
\usepackage[colorlinks=true, citecolor=purple, urlcolor=blue]{hyperref}
\usepackage{comment}
\makeatletter
\newcommand{\setword}[2]{%
  \phantomsection
  #1\def\@currentlabel{\unexpanded{#1}}\label{#2}%
}
\makeatother
\usepackage{comment}
\usepackage{enumitem}
\usepackage{graphicx} 
\usepackage[babel]{microtype}  
\usepackage{amsmath,amssymb,amsthm,bm,amsfonts,mathrsfs,bbm} 

\usepackage{xspace}  
\usepackage{pgf,tikz}
\usepackage{xcolor}
\usepackage{multirow}
\usepackage{array}
\usepackage{bigstrut}
\usepackage{braket}
\usepackage{color}
\usepackage{natbib}
\usepackage{multirow}
\usepackage{mathtools}
\usepackage{float}
\usepackage[caption = false]{subfig}
\usepackage{xcolor,colortbl}
\usepackage{color}

\newcommand{\be}{\begin{equation}}
\newcommand{\ee}{\end{equation}}
\newcommand{\ba}{\begin{eqnarray}}
\newcommand{\ea}{\end{eqnarray}}

\newcommand{\tr}{\operatorname{Tr}}

\newtheorem{theorem}{Theorem}

\newtheorem{definition}{Definition}
\newtheorem{proposition}{Proposition}
\newtheorem{observation}{Observation}

\newtheorem{remark}{Remark}
\newtheorem{lemma}{Lemma}

\def\d{{\rm d}}

\def\>{\rangle}
\def\<{\langle}

\usepackage{centernot}
\usepackage{subfig}
\usepackage{filecontents}

\providecommand{\ket}[1]{| #1{\rangle}}
\providecommand{\bra}[1]{\langle #1|}

\providecommand{\dacb}[2]{\{\hspace{-.12cm}\{#1#2\}\hspace{-.12cm}\}}
\providecommand{\dcb}[2]{[\hspace{-.05cm}[#1#2]\hspace{-.05cm}]}

\begin{document}

\title{No Absolute Hierarchy of Quantum Complementarity}

\author{Kunika Agarwal}
\thanks{These authors contributed equally and share the credit of first author.}
\affiliation{Department of Physics of Complex Systems, S. N. Bose National Center for Basic Sciences, Block JD, Sector III, Salt Lake, Kolkata 700106, India.}

\author{Sahil Gopalkrishna Naik}
\thanks{These authors contributed equally and share the credit of first author.}
\affiliation{Department of Physics of Complex Systems, S. N. Bose National Center for Basic Sciences, Block JD, Sector III, Salt Lake, Kolkata 700106, India.}

\author{Ananya Chakraborty}
\affiliation{Department of Physics of Complex Systems, S. N. Bose National Center for Basic Sciences, Block JD, Sector III, Salt Lake, Kolkata 700106, India.}

\author{Guruprasad Kar}
\affiliation{Physics and Applied Mathematics Unit, 203 B.T. Road Indian Statistical Institute Kolkata, 700108, India.}

\author{Ram Krishna Patra}
\affiliation{HUN-REN Institute for Nuclear Research, P.O. Box 51, H-4001 Debrecen, Hungary.}

\author{Snehasish Roy Chowdhury}
\affiliation{Physics and Applied Mathematics Unit, 203 B.T. Road Indian Statistical Institute Kolkata, 700108, India.}

\author{Manik Banik}
\affiliation{Department of Physics of Complex Systems, S. N. Bose National Center for Basic Sciences, Block JD, Sector III, Salt Lake, Kolkata 700106, India.}

\begin{abstract}
Bohr’s principle of complementarity, prohibiting simultaneous access to certain physical properties within a single experimental arrangement, is considered to be a defining feature of quantum mechanics. It is commonly viewed as inducing an intrinsic hierarchy among incompatible observables: some sets of quantum properties are fundamentally more incompatible than others, as quantified by the maximal sharpness permitting their joint measurement. We show that this hierarchy ceases to be absolute in the multi-copy regime. Analyzing qubit spin observables, we prove a No-Comparison Theorem establishing that no global ordering of incompatible observable sets is preserved across all finite-copy configurations. In particular, two sets of observables can exhibit reversed complementarity ordering depending solely on whether the available resources are arranged as identical copies or as parallel–antiparallel pairs. Thus, the degree of quantum incompatibility is not an intrinsic property of observables alone but depends on the global configuration of the prepared quantum probes. Our results uncover a configuration-dependent structure of complementarity, reveal a subtle role of entanglement in shaping the structure of measurement limitations, and call for a reassessment of quantum information protocols under finite resources.  
\end{abstract}


\maketitle	
{\it Introduction.--} Quantum mechanics fundamentally departs in various ways from our day-to-day classical intuitions. This departure is most famously articulated in Bohr’s principle of complementarity, which asserts that certain physical properties cannot be unambiguously revealed within a single experimental arrangement \cite{Bohr1928,Bohr1950,Englert1995}. Its physical meaning is vividly illustrated by the trade-off between path information and interference visibility in the double-slit experiment \cite{Feynman1965,Scully1991}, and more formally by the impossibility of jointly measuring non-commuting observables such as position and momentum \cite{Arthurs1965,Weston2013} or spin components along different directions \cite{Davies1976,Lahti1980,Busch1985}. While Bohr viewed Heisenberg’s uncertainty relations \cite{Heisenberg1927} as an expression of complementarity, and Pauli offered his own interpretations \cite{Pauli1994}, a precise operational formulation of complementarity emerges within the modern theory of joint measurements \cite{Mittelstaedt1987,Busch1996,Busch1986,Heinosaari2016,Guhne2023}. This framework rigorously characterizes when observables can be assessed simultaneously and with what degree of precision. Unlike classical observables, which can in principle be measured concurrently without restriction, incompatible quantum observables resist such joint determination. Generalized measurements, described by positive-operator valued measures (POVMs) \cite{Kraus1983}, permit approximate joint measurements at the cost of reduced sharpness. The maximal sharpness compatible with a single measurement device provides a quantitative measure of complementarity and naturally induces a hierarchy among sets of observables \cite{Busch1986,Heinosaari2016,Guhne2023}: some collections of quantum properties appear intrinsically more incompatible than others. This ordering is widely regarded as a structural feature of quantum theory.

Here we demonstrate that this hierarchy is not absolute. We show that the degree of complementarity between distinct sets of observables ceases to admit a universal ordering once one enters the multi-copy regime. Specifically, we investigate joint measurability when multiple copies of a quantum system are available and when the probing states are arranged in different finite configurations. We prove a No-Comparison Theorem establishing that no state-independent ordering of incompatible observable sets is preserved across all finite-copy configurations. In particular, two sets of observables can exhibit reversed complementarity ordering depending solely on whether the available resources are arranged as identical copies or as parallel--antiparallel pairs. Thus, the incompatibility of observables is not an intrinsic property of the observables alone, but a relational feature emerging from the global configuration of finite quantum resources. The structure of complementarity, long assumed to be fixed, becomes configuration-dependent. Moreover, we identify a previously unrecognized role of entanglement in shaping measurement limits themselves, rather than merely being constrained by them. These findings reveal that quantum complementarity does not admit a universal hierarchy and establish its fundamentally relational character in the finite-resource regime.

{\it Quantifying Complementarity.--} The measurable properties of a quantum system $S$ are represented by Hermitian operators acting on the Hilbert space $\mathcal{H}_S$ associated with it. The standard procedure of measuring such an observable is described by a set of projectors $\mathrm{M}_{\text{proj}}\equiv\{ \mathrm{P}_{\psi_i} := \ket{\psi_i}\bra{\psi_i}\}_{i=1}^d$ satisfying $\sum_i \mathrm{P}_{\psi_i} = \mathbf{I}_d$; with $\mathbf{I}_d$ denoting the identity operator on $\mathcal{H}_S$ and $d=\text{dim}(\mathcal{H}_S)$. A generic quantum state is represented by density operator, positive semi-definite operators having unit trace, $\rho \in \mathcal{D}(\mathcal{H}_S)$; $\mathcal{D}(\mathcal{H}_S)$ denotes the set of all density operators acting on $\mathcal{H}_S$. The probability of outcome `$i$' when a measurement $\mathrm{M}_{\text{proj}}$ is performed on a quantum system prepared in state $\rho$ is given by the Born rule, $p(i|\rho, \mathrm{M}_{\text{proj}}) = \operatorname{tr}[\rho \mathrm{P}_{\psi_i}]$. The more general framework of POVMs is essential for our analysis. A POVM is a set of positive semidefinite operators $\mathrm{M}\equiv\{\Pi_i\}_{i=1}^n$ satisfying $\sum_i \Pi_i = \mathbf{I}_d$~\cite{Kraus1983}. POVMs encompass all physically realizable measurements and are crucial for analyzing joint measurability. We formalize our investigation in qubits, a two-level quantum system -- spin-1/2 system being the canonical example. State of a qubit allows Bloch sphere representation $\rho_{\vec{m}}=\tfrac{1}{2}(\mathbf{I}_2+\vec{m}.\sigma)$, where $\vec{m}\in\mathbb{R}^3$ with $|\vec{m}|\le1$. An unsharp qubit spin observable, a specific form of POVM, along the space direction $\hat{n}$, is given by $\sigma_{\hat{n}}(\lambda)\equiv \big\{\mathrm{P}_{a\hat{n}}(\lambda):=\lambda\mathrm{P}_{a\hat{n}}+(1-\lambda)\tfrac{1}{2}\mathbf{I}_2~|~a=\pm1\big\}$, where $\mathrm{P}_{a\hat{n}}=\ket{a\hat{n}}\bra{a\hat{n}}$ is the eigenprojector corresponding to the eigenvalue $a\in\{\pm1\}$. The sharpness values $1$ and $0$ corresponding to the sharp (projective) measurement and the completely uninformative measurement, respectively. Joint measurability of unsharp such observables is defined by \cite{Busch1996,Busch1986,Heinosaari2016,Guhne2023}:
\begin{definition}\label{def1}
A set of $N$ distinct qubit spin observables $\mathcal{O}\equiv\{\sigma_{\hat{n}_r}~|~r=1,\cdots N\}$ is jointly measurable, up to a sharpness value $\lambda^{_{[1]}}_{\mathcal{O}}$, if there exists a single POVM $\mathcal{G}^{_{[1]}}\equiv\big\{\Pi^{_{[1]}}_{\mathbf{a}}~|~\sum_{\mathbf{a}}\Pi^{_{[1]}}_{\mathbf{a}}=\mathbf{I}_2\big\}$, with $\mathbf{a}:=a_1\cdots a_N \in \{+1,-1\}^N$, s. t. $\sum_{\mathbf{a}\setminus a_r} \tr\left[ \Pi^{_{[1]}}_{\mathbf{a}}\rho_{\vec{m}}\right]= \tr\left[ \mathrm{P}_{a_r\hat{n}_r}\big(\lambda^{_{[1]}}_{\mathcal{O}}\big) \rho_{\vec{m}} \right]~\forall~r,\rho_{\vec{m}}$, implying $\sum_{\mathbf{a}\setminus a_r} \Pi^{_{[1]}}_{\mathbf{a}}= \mathrm{P}_{a_r\hat{n}_r}\big(\lambda^{_{[1]}}_{\mathcal{O}}\big)$.   
\end{definition}
\noindent Here, $\sum_{\mathbf{a} \setminus a_r}$ denotes summation over all components of $\mathbf{a}$ except the $r$-th one. This definition captures the operational meaning of measuring properties simultaneously. However, for non-commuting observables, such as spin along different axes, this simultaneous access comes at a cost. Perfect, projective measurement is forbidden, and joint measurability is only possible with a reduced precision, quantified by a sharpness parameter $\lambda \in [0,1]$. The maximum sharpness $\lambda_{\mathcal{O}}$ for which a given set of observables $\mathcal{O}$ are jointly measurable provides a natural and operational measure of their complementarity: a lower maximum sharpness indicates a stronger incompatibility within the observables\cite{Wolf2009,Banik2013,Heinosaari2015,Banik2015}. This naturally induces a hierarchy among sets of observables. We can say a set of observables $\mathcal{O}$ exhibits a higher degree of complementarity than another set $\mathcal{O}'$, denoted $\mathcal{O} \succ \mathcal{O}'$, if the maximum sharpness parameter ensuring joint measurement of the observables in former set is less than that of the later, namely $\lambda_{\mathcal{O}} < \lambda_{\mathcal{O}'}$. For instance, two spin observables along mutually orthogonal space directions are jointly measurable up to the sharpness threshold $1/\sqrt{2}$, whereas for three mutually orthogonal spin observables the threshold further reduces to $1/\sqrt{3}$ \cite{Busch1986}. This establishes a seemingly fundamental ordering, suggesting that some sets of properties are intrinsically more incompatible than others. 
\begin{figure}[t!]
\centering
\includegraphics[width=1\linewidth]{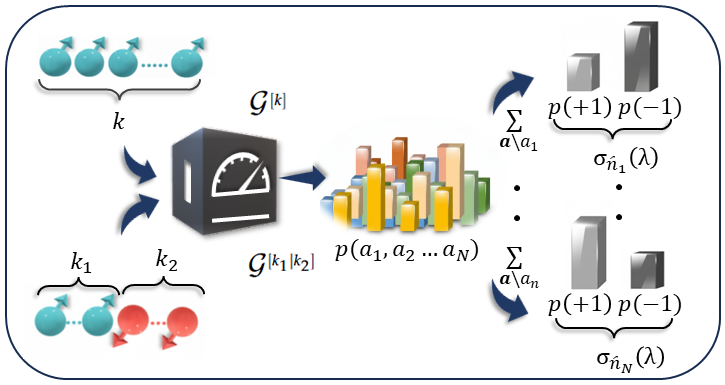}
\caption{Multi-copy joint measurability: A device $\mathcal{G}^{[k]}$ accesses $k$-copy of identical qubit states in each experimental run, and yields the outcome statistics $\{p({\bf a})\}_{{\bf a}}$, with $\mathbf{a}:=a_1\cdots a_N \in \{+1,-1\}^N$. The device jointly measures a set $\mathcal{O}$ of $N$ distinct observables up to an optimal sharpness value $\lambda~(=\lambda^{[k]}_{\mathcal{O}})$ if the outcome statistics $\{p({\bf a})\}_{\bf a}$ can be processed to reproduce measurement statistics of each observable in $\mathcal{O}$ with that effective sharpness. The device $\mathcal{G}^{[k_1|k_2]}$, on the other hand, accesses $k_1$ copies of a qubit state along with $k_2$ copies of its spin-flipped version and ensures joint measurability up to a sharpness value $\lambda~(=\lambda^{[k_1|k_2]}_{\mathcal{O}})$.}
\label{figJM}
\end{figure}

{\it Multi-copy Joint Measurability.--} The paradigm shifts when we enter the multi-copy regime, where a measurement device probes $k$ copies of a quantum state in each experimental run (see Figure \ref{figJM}). In this setup, a set of observables is $k$-copy jointly measurable if there exists a single, potentially entangled, measurement on the joint Hilbert space $\mathcal{H}^{\otimes k}$ whose outcome statistics can be post-processed to reproduce the statistics of each observable with an effective sharpness $\lambda^{[k]}$ \cite{Carmeli2016}.
\begin{definition}\label{def2}
The set $\mathcal{O}$ is jointly measurable on the configuration $[k]$, up to a sharpness value $\lambda^{_{[k]}}_{\mathcal{O}}$, if there exists a single POVM $\mathcal{G}^{_{[k]}}\equiv\big\{\Pi^{_{[k]}}_{\mathbf{a}}~|~\sum_{\mathbf{a}}\Pi^{_{[k]}}_{\mathbf{a}}=\mathbf{I}^{\otimes k}_2\big\}$, s.t. $\sum_{\mathbf{a}\setminus a_r} \tr\left[ \Pi^{_{[k]}}_{\mathbf{a}} \rho_{\vec{m}}^{\otimes k}\right] = \tr\left[ \mathrm{P}_{a_r\hat{n}_r}\big(\lambda^{_{[k]}}_{\mathcal{O}}\big) \rho_{\vec{m}} \right]~\forall~r,~\rho_{\vec{m}}$.  
\end{definition}
\noindent Notably, unlike the single-copy case, the joint measurability conditions in multi-copy case do not boils down to a set of operator equalities. One can also consider an asymmetric configurations `$[k_1|k_2]$' wherein the measurement device is fed with $k_1$ copies of a qubit state and $k_2$ copies of their spin-flipped counterpart.
\begin{definition}\label{def3}
The set $\mathcal{O}$ is jointly measurable on the configuration $[k_1|k_2]$, up to a sharpness value $\lambda^{_{[k_1|k_2]}}_{\mathcal{O}}$, if there exists a single POVM $\mathcal{G}^{_{[k_1|k_2]}}\equiv\big\{\Pi^{_{[k_1|k_2]}}_{\mathbf{a}}~|~\sum_{\mathbf{a}}\Pi^{_{[k_1|k_2]}}_{\mathbf{a}}=\mathbf{I}^{\otimes (k_1+k_2)}_2\big\}$, s.t. $\sum_{\mathbf{a} \setminus a_r} \tr\left[ \Pi^{_{[k_1|k_2]}}_{\mathbf{a}}\hspace{-.1cm} \left( \rho_{\vec{m}}^{\otimes k_1} \otimes \rho_{-\vec{m}}^{\otimes k_2} \right) \right]= \tr\left[\mathrm{P}_{a_r\hat{n}_r}\big(\lambda^{_{[k_1|k_2]}}_{\mathcal{O}}\big) \, \rho_{\vec{m}} \right]~\forall~r,~\rho_{\vec{m}}$.    
\end{definition}
\noindent The simplest such configuration, the $[1|1]$ antiparallel pair, is known to encode more quantum information than the parallel configuration of two identical copies \cite{Gisin1999}. Recent research now reveals a further advantage: the antiparallel pair can also achieve higher joint measurability threshold for certain sets of complementary observables as compared to the parallel arrangement \cite{Patra2026}.

{\it Configuration-Dependent Complementarity (CDC).--} We investigate two configurations: the symmetric configuration $[k]$, and the asymmetric configuration $[k_1|k_2]$. While prior research establishes advantage of  asymmetric configuration over the symmetric one for certain observables, we now demonstrate a even more profound phenomenon: the very hierarchy of complementarity between different sets of observables undergoes a dramatic reversal when comparing these configurations. To concretely demonstrate this reversal, we first analyze two geometrically symmetric sets of spin observables: 
\begin{itemize}[itemsep=-.25cm, parsep=.35cm, topsep=2pt, leftmargin=.3cm]
\item $\mathtt{SyTri}$:- Contains three spin observables with measurement axes arranged symmetrically in the equatorial plane forming an equilateral triangle;
\item $\mathtt{SyTet}$:- Contains four spin observables with measurement axes oriented along the vertices of a regular tetrahedron.
\end{itemize}
\noindent We also consider two configurations: three identical copies of a qubit state -- the configuration $[3]$, and the anti-parallel configuration -- the configuration $[1|1]$. With this we can now state our No-Comparison theorem. 
\begin{figure}[t!]
\centering
\includegraphics[width=1\linewidth]{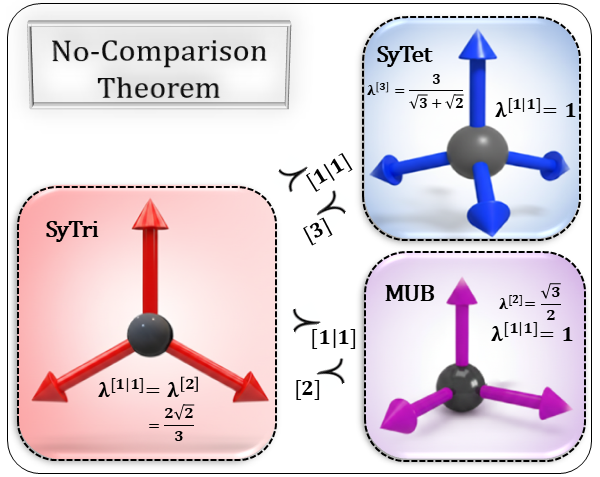}
\caption{Configuration-Dependent Complementarity: The triangular set $\mathtt{SyTri}$ is perfectly joint measurable on three identical qubit, while the tetrahedral set $\mathtt{SyTet}$ is not. Conversely, on anti-parallel configuration $\mathtt{SyTet}$ is perfectly jointly measurable but $\mathtt{SyTri}$ is not. For the $\mathtt{MUB}$ observables $\lambda^{_{[2]}}_{\mathtt{MUB}}=\sqrt{3}/2$ \cite{Carmeli2016}, whereas $\lambda^{_{[1|1]}}_{\mathtt{MUB}}=1$ \cite{Patra2026}. On the other hand, for the observable set $\mathtt{SyTri}$ we have  $\lambda^{_{[2]}}_{\mathtt{SyTri}}=\lambda^{_{[1|1]}}_{\mathtt{SyTri}}=2\sqrt{2}/3$, thereby supporting the thesis of No-Comparison theorem.}
\label{figNo-C}
\end{figure}

\begin{theorem}\label{theo1}
[No-Comparison] The spin observables $\mathtt{SyTri}$ are perfectly jointly measurable on three identical copies of a qubit state but not in the anti-parallel configuration, whereas $\mathtt{SyTet}$ observables show the opposite behavior. 
\end{theorem}
\noindent The proof employs two impossibility arguments as formalized in the following two Propositions with their proof detailed in Appendix. 
\begin{proposition}\label{prop1}
No set $\mathcal{O}$ of $N$ distinct spin observables is perfectly jointly measurable on $[N-1]$ configuration, i.e., the optimal sharpness value $\lambda^{_{[N-1]}}_{^{\mathcal{O}}}< 1$.    
\end{proposition}
\noindent While Proposition~\ref{prop1} establishes a strict limitation for perfect joint measurability $N$ distinct spin observables on $(N-1)$ identical copies of qubit, here we highlight that access to multiple copies in a single experimental run offers a quantifiable advantage for joint measurement, as measured by an increased sharpness parameter \cite{Carmeli2016}. The authors in \cite{Carmeli2016} introduced `index of incompatibility' defined as the minimal number of copies that is needed in order to make a given set of observables compatible. They also posed the question whether there exist sets of $n$ observables whose index of incompatibility is $n$. Our Proposition~\ref{prop1} answers this question in full generality by showing that, for any set of $n$ distinct spin observables, the index of incompatibility is indeed $n$. Notably, a result analogous to Proposition~\ref{prop1} does not hold in general for configuration $[k_1|k_2]$ even when $k_1+k_2<N$. For instance, the $\mathtt{MUB}$ observables can be perfectly joint measurable on anti-parallel configuration $[1|1]$ \cite{Patra2026}. Nonetheless, in the next we establish a generic result when the observables satisfy a particular geometric constraint.
\begin{proposition}\label{prop2}
For any set of $N$ distinct spin observables $\mathcal{O}$ with the measurement directions in a fixed equatorial plane of the Bloch sphere, the optimal joint measurability sharpness on two-copy state  is configuration-independent, i.e., $\lambda_{\mathcal{O}}^{_{[2]}} = \lambda_{\mathcal{O}}^{_{[1|1]}}$.
\end{proposition}
\noindent Proposition~\ref{prop1} and Proposition ~\ref{prop2} ensure $\lambda^{_{[3]}}_{^{\mathrm{SyTet}}}<1$ as well as $\lambda^{_{[1|1]}}_{^{\mathrm{SyTri}}}=\lambda^{_{[2]}}_{^{\mathrm{SyTri}}}<1$. On the other hand, we have $\lambda^{_{[3]}}_{^{\mathrm{SyTri}}}=1$ as the set $\mathrm{SyTri}$ contains three spin observables. We then design a specific entangled POVM acting on $2$-qubit Hilbert space that when applied to the antiparallel configuration perfectly recovers measurement statistics of all the four tetrahedral observables simultaneously, thereby ensuring $\lambda^{[1|1]}_{\mathrm{SyTet}}=1$ (see Appendix), which thus proves the claim of our No-comparison theorem. The stark reversal established in No-comparison theorem thus  definitively demonstrates that the question ``Which set is more complementary?" admits no configuration-independent answer -- the hierarchy is fundamentally relational rather than absolute. Notably, the generality of Propositions~\ref{prop1} and ~\ref{prop2} allow us to replace the observables set $\mathtt{SyTri}$ in the No-comparison theorem by any set of three distinct equatorial spin observables.

Remarkably, this incomparability phenomenon extends beyond the regime of perfect joint measurability. Even when exact simultaneous measurement is impossible, the sharpness hierarchy exhibits the same configuration-dependent reversal. For instance, consider mutually unbiased bases ($\mathtt{MUB}$) set consisting of three spin observables along mutually orthogonal space directions. As we argue, on the anti-parallel configuration $[1|1]$, $\mathtt{MUB}$ observables are less complementary than triangular ones, i.e., $\lambda_{\mathtt{MUB}}^{[1|1]} > \lambda_{\mathtt{SyTri}}^{[1|1]}$, while for two identical copies $[2]$, this relationship inverts, with $\mathtt{MUB}$ becoming more complementary than the triangular set, namely $\lambda_{\mathtt{MUB}}^{[2]} < \lambda_{\mathrm{SyTri}}^{[2]}$. This secondary reversal, illustrated in Figure \ref{figNo-C}, reinforces our central thesis: no universal ordering of complementarity exists independent of the quantum resource configuration.
\begin{figure}[t!]
\centering
\includegraphics[width=.7\linewidth]{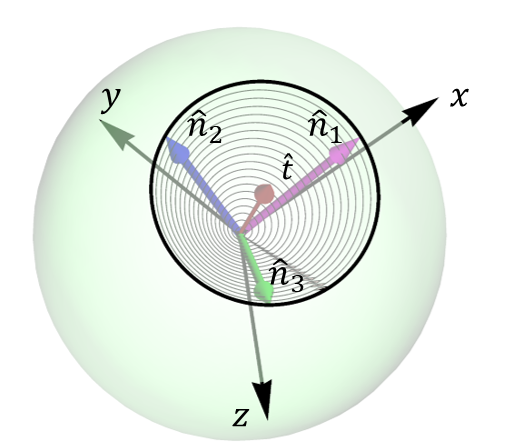}
\caption{Three symmetrically arranged spin observables: Symmetric configuration of three spin observables $\mathcal{O}_\theta \equiv \{\sigma_{\hat{n}_i}\}_{i=1}^3$. Each observable's direction $\hat{n}_i$ makes an angle $\theta \in (0, \pi)$ with the reference axis $\hat{t} = (1,1,1)^{\text{T}}/\sqrt{3}$. The angle between any two distinct observables is $\varphi = \cos^{-1}\left((3\cos^2\theta - 1)/2\right)$.}
\label{figs1}
\end{figure}
\begin{figure}[b!]
\centering
\includegraphics[width=1\linewidth]{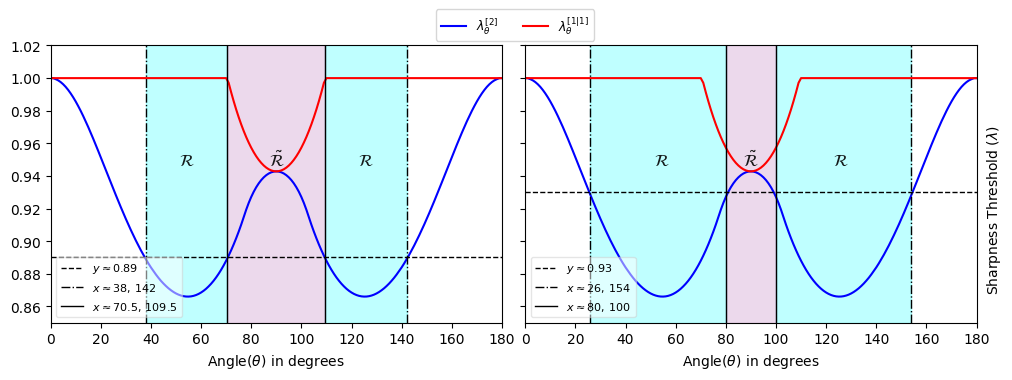}
\caption{For any triple of spin observables $\mathcal{O}_{\tilde{\theta}}$ with $\tilde{\theta}\in \tilde{\mathcal{R}}$ and any triple of spin observables $\mathcal{O}_\theta$ with $\theta \in \mathcal{R}$, we observe that $\lambda^{_{[2]}}_{\mathcal{O}_{\tilde{\theta}}} > \lambda^{_{[2]}}_{\mathcal{O}_\theta}$ while $\lambda^{_{[1|1]}}_{\mathcal{O}_{\tilde{\theta}}} < \lambda^{_{[1|1]}}_{\mathcal{O}_\theta}$, thereby re-establishing the core thesis of the No-Comparison theorem. As the region $\tilde{\mathcal{R}}$ decreases, the region $\mathcal{R}$ increases, a fact evident from the two sub-figures. The red and blue curves coincide at $\theta = 0$, $\pi/2$, and $\pi$. At $\theta = 0$ and $\theta = \pi$, the set $\mathcal{O}_\theta$ contains only a single measurement. The case $\theta = \pi/2$ is particularly interesting, as the three distinct observables lie in the great plane orthogonal to the $\hat{t}:=\tfrac{1}{\sqrt{3}}(1,1,1)$ axis. The relationship between the sharpness thresholds for parallel and anti-parallel configurations, in this case, is consistent with the claim of Proposition~\ref{prop2}}
\label{figs2}
\end{figure}

The asymmetric $[1|1]$ configuration enables perfect joint measurability of the $\mathtt{SyTet}$ observables using just $2$ qubits, whereas Proposition~\ref{prop1} implies that $4$ identical qubits are required in the symmetric configuration. This represents a $2$-qubit advantage per $\mathtt{SyTet}$ set. Consider a set $\textbf{n-}\mathtt{SyTet}$, comprising the union of $n$ distinct sets of 4 observables oriented along the vertices of different regular tetrahedrons, thereby containing $4n$ distinct spin observables. According to Proposition \ref{prop1} perfect joint measurability of the observables in $\textbf{n-}\mathtt{SyTet}$ requires the configuration $[4n]$ when considering symmetric configuration, whereas the same is possible in $[n|n]$ asymmetric configuration, thereby scaling the advantage to $2n$ qubits. Whether more sophisticated entangled measurements across multiple qubit could yield further advantage remains an open question for future investigation. At this stage, an intriguing question is whether the reversible behavior of complementarity occur without invoking spin-flipped counterparts in the configurations? We believe that for two sets of incompatible spin observables, $\mathcal{O}$ and $\mathcal{O}'$, with $|\mathcal{O}| \le |\mathcal{O}'|$, if $\lambda^{[1]}_{\mathcal{O}} \ge \lambda^{[1]}_{\mathcal{O}'}$ then $\lambda^{[k]}_{\mathcal{O}} \ge \lambda^{[k]}_{\mathcal{O}'}$ for all $k$. Although this claim is currently supported only by numerical evidence, a formal proof or counterexample would yield deeper insight into quantum complementarity in finite-copy symmetric configurations.

{\it Generic Nature of CDC.--} The configuration-dependent reversal of complementarity established by our No-Comparison Theorem is not a discrete phenomenon limited to specific sets like $\mathtt{SyTri}$ and $\mathtt{SyTet}$. To demonstrate its generality, we analyze a parametric family of three spin observables $\mathcal{O}_{\theta} \equiv \{\sigma_{\hat{n}_i}\}_{i=1}^3$, whose measurement directions $\{\hat{n}_1, \hat{n}_2, \hat{n}_3\}$ are symmetrically arranged with identical pairwise angles: $\hat{n}_i \cdot \hat{n}_j = \cos\varphi \quad \text{for } i \neq j$. This symmetry is captured by the parameterization: $\hat{n}_1= (\alpha, \beta, \beta)^{\mathrm{T}},~\hat{n}_2 = (\beta, \alpha, \beta)^{\mathrm{T}},~\hat{n}_3 = (\beta, \beta, \alpha)^{\mathrm{T}}$, where  $\alpha:= \tfrac{1}{3}\big(\sqrt{3}\cos\theta + \sqrt{6}\sin\theta\big)$ and $\beta:= \tfrac{1}{6}\big(2\sqrt{3}\cos\theta - \sqrt{6}\sin\theta\big)$, with $\theta := \cos^{-1}(\hat{t} \cdot \hat{n}_i) \in (0, \pi)$ and $\hat{t} = \tfrac{1}{\sqrt{3}}(1,1,1)^{\mathrm{T}}$ (see Figure~\ref{figs1}). Accordingly, we have $\varphi=\cos^{-1}\big((3\cos^2\theta-1)/2\big)$. While $\theta=\cos^{-1}\big(1/\sqrt{3}\big)$ corresponds to the $\mathtt{MUB}$ observables, $\theta=\pi/2$ represents the $\mathtt{SyTri}$ set. Our objective is to determine the optimal sharpness thresholds $\lambda^{[2]}_\theta$ and $\lambda^{[1|1]}_\theta$ that permit joint measurability of $\mathcal{O}_{\theta}$ on the parallel ($[2]$) and anti-parallel ($[1|1]$) configurations, respectively. Leveraging {\bf Result~1} and {\bf Result~2}, stated in Appendix, this optimization can be formulated as semidefinite programs (SDPs):
\begin{align}
&\underline{\text{Parallel Configuration $[2]$}}\nonumber\\
&\text{\bf Maximize:} \quad  \lambda^{[2]}_\theta \nonumber \\
&\text{Subject to:}~\Pi^{[2]}_{\mathbf{a}} \ge 0, ~\sum_{\mathbf{a}} \Pi^{[2]}_{\mathbf{a}} = \mathbf{I}_2 \otimes \mathbf{I}_2,~\&~\text{for~all}~r \nonumber \\
&\sum_{\mathbf{a} \setminus a_r} \Pi^{[2]}_{\mathbf{a}} = \tfrac{1}{2} \left( \mathrm{P}_{a_r \hat{n}_r}(\lambda^{[2]}_{\theta}) \otimes \mathbf{I}_2 + \mathbf{I}_2 \otimes \mathrm{P}_{a_r \hat{n}_r}(\lambda^{[2]}_{\theta}) \right).\nonumber\\
&\underline{\text{Anti-parallel Configuration $[1|1]$}}\nonumber\\
&\text{\bf Maximize:} \quad \lambda^{[1|1]}_\theta \nonumber \\
&\text{Subject to:}~ \Pi^{[1|1]}_{\mathbf{a}}\ge 0, ~ \sum_{\mathbf{a}} \Pi^{[1|1]}_{\mathbf{a}} = \mathbf{I}_2 \otimes \mathbf{I}_2,,~\&~\text{for~all}~r \nonumber \\
&\sum_{\mathbf{a} \setminus a_r} \Pi^{[1|1]}_{\mathbf{a}} = \tfrac{1}{2} \left( \mathrm{P}_{a_r \hat{n}_r}(\lambda^{[1|1]}_{\theta}) \otimes \mathbf{I}_2 + \mathbf{I}_2 \otimes \mathrm{P}_{-a_r \hat{n}_r}(\lambda^{[1|1]}_{\theta})\right).\nonumber
\end{align}
Solving these SDPs reveals configuration-dependent complementarity behavior for specific ranges of the parameter $\theta$, as depicted in Figure~\ref{figs2}.

{\it Discussion.--} Our No-Comparison Theorem establishes a structural revision of complementarity: in the finite-resource regime, quantum complementarity does not admit an absolute hierarchy. The ordering of incompatible observables depends on how quantum systems are configured. This insight connects two foundational developments of the last century—the formalization of information theory \cite{Shannon1948,Goldstine1961} and its unification with quantum mechanics \cite{Bennett1998}. Information, now understood as a physical primitive \cite{Landauer1961,Bennett1973,Zurek1989}, plays an operational role in shaping physical limits \cite{Wheeler,Deutsch2004}. The qubit, enabled by superposition yet constrained by no-cloning \cite{Wootters1982}, raises a central question: what information can be extracted from finite copies, and how does their arrangement affect measurement limits? We show that complementarity is not intrinsic to observables alone but emerges from the global configuration of finite resources.

A parallel can be drawn with entanglement theory. Entanglement, the archetypal non-classical correlation \cite{Einstein1935,Schrodinger1936,Bell1964}, is a resource under local operations and classical communication \cite{Horodeck2009} and central to composite-system structure \cite{Barnum2010,Torre2012,Naik2022,Patra2023}. Entangled states are only partially ordered: certain states are fundamentally incomparable \cite{Bennett1996,Nielsen1999,Vidal1999}. For instance, the Bell state $|\phi\rangle=(|00\rangle+|11\rangle)/\sqrt{2}$ and the two-qutrit state $|\psi\rangle=(|\phi\rangle+3|22\rangle)/\sqrt{10}$ reverse their ordering depending on whether entanglement entropy or Schmidt rank is used. Such inequivalence persists even asymptotically for nonlocal correlations \cite{Ghosh2024}. By contrast, the incomparability we identify concerns physical properties rather than states. While entanglement with a quantum memory modifies uncertainty bounds \cite{Berta2010}, here entanglement appears in the optimal measurement strategy itself. The reversed complementarity ordering between identical copies and antiparallel pairs reflects the configuration sensitivity of this measurement-induced entanglement.

Operationally, no single finite configuration universally optimizes information extraction. Whether identical copies or antiparallel arrangements are advantageous depends on the observables considered. This has implications for precision metrology \cite{Giovannetti2011} and quantum memory design \cite{Lvovsky2009}, where finite resources are unavoidable. From a resource-theoretic viewpoint, measurement incompatibility is increasingly treated as an operational resource \cite{Carmeli2019,Skrzypczyk2019,Buscemi2020}. The configuration-dependent complementarity uncovered here adds a structural dimension to this program, indicating that incompatibility cannot be fully characterized without accounting for global resource arrangement.

\begin{acknowledgements}
\noindent{\bf Acknowledgement}: We thankfully acknowledge the useful discussions with Mir Alimuddin, Trina Chakraborty, Pratik Ghosal, Subhendu B. Ghosh, Amit Mukerjee, Asutosh Rai, and Arup Roy. KA acknowledges support from the CSIR project 09/0575(19300)/2024-EMR-I. SGN acknowledges support from the CSIR project 09/0575(15951)/2022-EMR-I. SRC acknowledges support from University Grants Commission, India (reference no. 211610113404). MB acknowledges funding support from the National Quantum Mission (NQM), an initiative of the Department of Science and Technology, Govt. of India.
\end{acknowledgements}


%

\onecolumngrid
\section{Appendix}
\section{Proof of Proposition \ref{prop1}}
\noindent We begin by recalling the definition of permutation-symmetric subspace for $N$ qubits.
\begin{definition}\label{def4}
The permutation-symmetric subspace $\mathrm{Sym}^N(\mathbb{C}^2)\subset(\mathbb{C}^2)^{\otimes N}$ for an $N$-qubit systems is the subspace that remains invariant under all the permutations of $N$ qubits. A canonical basis for this subspace is the normalized Dicke states \cite{Dicke1954} 
\begin{align}
\ket{\widetilde{D}_N^{(r)}}:= \frac{1}{\sqrt{\Big(\substack{N\\r}\Big)}}
\sum_{\substack{\vec{m} \in \{0,1\}^N \\ |\vec{m}| = r}} 
\ket{\vec{m}},
\quad r\in\{0,1,\dots,N\},\nonumber  
\end{align}
where the Hamming weight $|\vec{m}|:= \sum_{i=1}^N m_i$ denotes the number of qubits in the excited state $\ket{1}$.    
\end{definition}
\noindent As it is evident, for the $N$-qubit system dimension of $\mathrm{Sym}^N(\mathbb{C}^2)$ is $(N+1)$. Each Dicke state $\ket{\widetilde{D}_N^{(r)}}$ represents a symmetric superposition of all computational basis states containing exactly $r$ excitations. For a system of identically prepared qubit states, the joint state always lies in this subspace. This plays a crucial role in our analysis of joint measurability on symmetric configuration. Before proceeding to the proof of Proposition~\ref{prop1}, we first present a lemma required for our analysis.
\begin{lemma}\label{lemma1}
Let $\{\ket{\psi_i}\}_{i=0}^{N}\subset\mathbb{C}^2$ be $N+1$ distinct qubit states. Then the set of their $N$-fold tensor powers $\left\{\ket{\psi_i}^{\otimes N}\right\}_{i=0}^{N}$ is linearly independent and spans the symmetric subspace $\mathrm{Sym}^N(\mathbb{C}^2)$.    
\end{lemma}
\begin{proof}
Since the states  $\ket{\psi_i}^{\otimes N}$'s are invariant under the permutations, they lie in the symmetric subspace $\mathrm{Sym}^N(\mathbb{C}^2)$. We thus need to prove that $\mathrm{Span}\left\{\ket{\psi_i}^{\otimes N}\mid_{i=0}^{N}\right\}=\mathrm{Sym}^N(\mathbb{C}^2)$. Expressing $\ket{\psi_i}$ in computational basis as $\ket{\psi_i}= \alpha_i \ket{0}+ \beta_i \ket{1}$ with $\alpha_i, \beta_i \in \mathbb{C}$, we get
\begin{align}
\ket{\psi_i}^{\otimes N} 
&= (\alpha_i \ket{0}+ \beta_i \ket{1})^{\otimes N}\nonumber\\
&=\alpha_i^N\ket{0}^{\otimes N}+\alpha_i^{N-1}\beta_i~\sum_{\text{Perm}}\left(\ket{0}^{\otimes{N-1}}\otimes\ket{1}\right)+\alpha_i^{N-2}\beta^2_i~\sum_{\text{Perm}}\left(\ket{0}^{\otimes{N-2}}\otimes\ket{1}^{\otimes 2}\right)+\cdots+\beta_i^N\ket{1}^{\otimes N}\nonumber\\
&=\sum_{r=0}^{N} \alpha_i^{N-r}\beta_i^{r}~\sum_{\substack{\vec{m}\in\{0,1\}^N \\ |\vec{m}|=r}}\hspace{-.1cm}\ket{\vec{m}}= \sum_{r=0}^{N} \alpha_i^{N-r}\beta_i^{r}\, \ket{D_N^{(r)}},\nonumber
\end{align}
where $\ket{D_N^{(r)}}$ denotes the unnormalized Dicke state with $r$ excitations. Factoring out $\alpha_i^{N}$ and defining $z_i = \beta_i / \alpha_i$, we can write
\begin{align}
\ket{\psi_i}^{\otimes N} \propto \sum_{r=0}^{N} z_i^{r}\, \ket{D_N^{(r)}}.\nonumber
\end{align}
Note that, we can always ensure $\alpha_i$'s to be nonzero by expressing $\ket{\psi_i}$'s in suitably chosen orthonormal basis of $\mathbb{C}^2$. For instance, let $N=2$ and $\{\ket{\psi_i}\}_{i=0}^{N}=\{\ket{0},\ket{1},\ket{+}\}$ then the states can be expressed in Hadamard basis. Thus in Dicke basis the state $\ket{\psi_i}^{\otimes N}$ is  proportional to the column vector $(1, z_i, z_i^2, \dots, z_i^N)^{\mathrm{T}}$. Consider the $(N+1)\times(N+1)$ matrix $\mathbf{V}$ with $i^{th}$ row being $(1, z_i, z_i^2, \dots, z_i^N)$, i.e.
\begin{align}
\mathbf{V}=
\begin{pmatrix}
1 & z_0 & z_0^2 & \cdots & z_0^N\\
1 & z_1 & z_1^2 & \cdots & z_1^N\\
\vdots & \vdots & \vdots & \ddots & \vdots\\
1 & z_N & z_N^2 & \cdots & z_N^N
\end{pmatrix}.\nonumber
\end{align}
This is the well known \textbf{Vandermonde} matrix \cite{Klinger1967} with determinant
\begin{align}
\det(\mathbf{V}) = \prod_{0 \le i < j \le N} (z_j - z_i).\nonumber    
\end{align}
Qubit states $\ket{\psi_i}$'s being distinct ensure $z_i$'s are distinct, and hence $\det({\bf V}) \neq 0$. Therefore, the rows of ${\bf V}$ are linearly independent, implying that the vectors 
$\left\{\ket{\psi_i}^{\otimes N}\right\}_{i=0}^N$ are linearly independent. This completes the proof.
\end{proof}

\begin{table}[t!]
\begin{tabular}{c|c|c|c|c|}
\hline
Local dim. & Proj. space & Sym. $N$-fold dim. & Veronese image & Max \# of independent states\\ \hline\hline
$2$ (qubit)& $\mathbb{P}^1$ & $N+1$ & Veronese curve & $N+1$\\ \hline
$3$ (qutrit)& $\mathbb{P}^2$ & $(N+1)(N+2)/2$ & Veronese surface & $(N+1)(N+2)/2$\\ \hline
$d$ (qudit)& $\mathbb{P}^{d-1}$ & $\big(\substack{N+d+1\\d-1}\big)$ & Veronese variety & $\big(\substack{N+d+1\\d-1}\big)$\\ \hline
\end{tabular}
\end{table}

\noindent With a bit of digress, here we note that Lemma~\ref{lemma1} can be seen as a special case of {\bf Veronese embedding} \cite{Harris1992}, which ensures that given $N+1$ distinct points $[\psi_i]\in\mathbb{P}^1$, the Veronese images ${\nu_N}\left([\psi_i]\right)\in\mathbb{P}^N$ will be linearly independent. With this we are now in a position to prove our Proposition~\ref{prop1}. \\

\noindent{\bf Proof of Proposition 1}
\begin{proof}
We proceed by contradiction. Let $\mathcal{O}\equiv\{\sigma_{\hat{n}_r}\}_{r=1}^{N}$ be a set of $N$ distinct spin observables. Assume, contrary to the claim of Proposition {\bf 1}, that there exists a joint measurement $\mathcal{G}^{[k]}\equiv\left\{\Pi^{[k]}_{\mathbf{a}}~|~\Pi^{[k]}_{\mathbf{a}}\ge0~\&~\sum_{\mathbf{a}}\Pi^{[k]}_{\mathbf{a}}=\mathbf{I}^{\otimes k}_2\right\}$ on $k~(= N-1)$ identical copies of the system that simultaneously measures all $N$ observables in $\mathcal{O}$, namely
\begin{align}
\sum_{\mathbf{a}\setminus a_r} &
\mathrm{Tr}\!\left[ \Pi^{[k]}_{\mathbf{a}}\rho_{\vec{m}}^{\otimes k} \right]
= \mathrm{Tr}\!\left[ \mathrm{P}_{a_r\hat{n}_r} \rho_{\vec{m}} \right],~~\forall~\rho_{\vec{m}}\in\mathcal{D}(\mathbb{C}^2).\nonumber
\end{align}
Here $\mathrm{P}_{a_r\hat{n}_r}=\ket{\psi_{a_r\hat{n}_r}}\bra{{\psi_{a_r\hat{n}_r}}}$ denotes the eigen projector of $\sigma_{\hat{n}_r}$ corresponding to the eigenvalue $a_r = \pm 1$. Perfect joint measurability demands reproducing the measurement statistic on all states $\rho_{\vec{m}}\in\mathcal{D}(\mathbb{C}^2)$, implying 
\begin{align}
\sum_{\mathbf{a}\setminus a_r}\mathrm{Tr}\left[ \Pi^{[k]}_{\mathbf{a}}\mathrm{P}_{-a_r\hat{n}_r}^{~\otimes k}\right]=\mathrm{Tr}\!\left[\mathrm{P}_{a_r\hat{n}_r}\mathrm{P}_{-a_r\hat{n}_r} \right]= 0~\forall~r.\nonumber 
\end{align}
Since $\Pi^{[k]}_{\mathbf{a}}\ge0$ for all $\mathbf{a}$, we thus have
\begin{align}
\mathrm{Tr}\left[ \Pi^{[k]}_{\mathbf{a}}\mathrm{P}_{-a_r\hat{n}_r}^{~\otimes k} \right] = 0~~ \forall~ r, \mathbf{a}.\nonumber    
\end{align}
Therefore, for every outcome string $\mathbf{a}\in\{+1,-1\}^N$, the operator $\Pi^{[k]}_{\mathbf{a}}$
is orthogonal to each $\mathrm{P}_{-a_r\hat{n}_r}^{~\otimes k}$ implying
\begin{align}
\mathrm{Support}\left\{\Pi^{[k]}_{\mathbf{a}}\right\}\perp\mathrm{Span}\left\{\ket{\psi_{-a_r\hat{n}_r}}^{\otimes k}\mid_{r=1}^N\right\}.\nonumber   
\end{align}
By Lemma {\bf 1}, we have 
\begin{align}
\mathrm{Span}\left\{\ket{\psi_{-a_r\hat{n}_r}}^{\otimes k}\mid_{r=1}^N\right\}=\mathrm{Sym}^{k}(\mathbb{C}^2)=\mathrm{Sym}^{N-1}(\mathbb{C}^2).\nonumber  
\end{align}
Therefore, for $\mathbf{a}\in\{+1,-1\}^N$, $\Pi^{[k]}_{\mathbf{a}}$ is supported on the orthogonal complement of $\mathrm{Sym}^{N-1}(\mathbb{C}^2)$ which forms a $2^{N-1}-N$ dimensional subspace of $(\mathbb{C}^2)^{\otimes N-1}$. This further implies $\sum_{\bf a}\Pi^{[k]}_{\mathbf{a}}<\mathbf{I}_2^{\otimes N-1}$, establishing that $\mathcal{G}^{[k]}$ is not a valid POVM on $(\mathbb{C}^2)^{\otimes N-1}$. This completes the proof.
\end{proof}

\section{Proof of Proposition \ref{prop2}}
\noindent Here we start by recalling two results from prior works that will be used in our proof.

\noindent{\bf Result 1} [{\it Theorem 1 of \cite{Carmeli2016}}]: A set of spin observables $\mathcal{O}= \{\sigma_{\hat{n}_r}\}_r$ is jointly measurable on two identical qubits (i.e. the parallel configuration $[2]$) up to the sharpness threshold $\lambda^{[2]}_{\mathcal{O}}$, if and only if there exists a POVM $\mathcal{G}^{[2]}_{\mathcal{O}}\equiv \{\Pi^{[2]}_{\bf a} ~|~\sum_{{\bf a}} \Pi^{[2]}_{\bf a} = \mathbf{I}_2^{\otimes 2}\}$, such that $\sum_{{\bf a} \setminus a_r} \Pi^{[2]}_{{\bf a}}=\tfrac{1}{2}\big(\mathrm{P}_{a_r\hat{n}_r}(\lambda^{[2]}_{\mathcal{O}}) \otimes \mathbf{I}_2 + \mathbf{I}_2 \otimes \mathrm{P}_{a_r\hat{n}_r}(\lambda^{[2]}_{\mathcal{O}})\big)~\forall~r$.\\

\noindent{\bf Result 2} [{\it Proposition 1 of \cite{Patra2026}}]: A set of spin observables $\mathcal{O}= \{\sigma_{\hat{n}_r}\}_r$ is jointly measurable on anti-parallel configuration (i.e. the $[1|1]$ configuration) up to the sharpness threshold $\lambda^{[1|1]}_{\mathcal{O}}$, if and only if there exists a POVM $\mathcal{G}^{[1|1]}_{\mathcal{O}}\equiv \{\Pi^{[1|1]}_{\bf a} ~|~\sum_{{\bf a}} \Pi^{[1|1]}_{\bf a} = \mathbf{I}_2^{\otimes 2}\}$, such that $\sum_{{\bf a} \setminus a_r} \Pi^{[1|1]}_{{\bf a}}=\tfrac{1}{2}\big(\mathrm{P}_{a_r\hat{n}_r}(\lambda^{[1|1]}_{\mathcal{O}}) \otimes \mathbf{I}_2 + \mathbf{I}_2 \otimes \mathrm{P}_{-a_r\hat{n}_r}(\lambda^{[1|1]}_{\mathcal{O}})\big)~\forall~r$.\\

\noindent The authors in \cite{Carmeli2016} established a more general result for the symmetric configuration $[k]$. For our purposes, we restate their finding specifically for the $k=2$ case, adapting it to the notation of our present work. With this we now proceed to prove the Proposition {\bf 2}.\\

\noindent{\bf Proof of Proposition 2}
\begin{proof}
Let $\mathcal{O}\equiv \{\sigma_{\hat{n}_r}\}$ be a set of $N$ distinct observables with all the measurement directions lying in a plane. Without loss of any generality, we can consider the plane to be $XZ$-plane. Assume there exists a joint measurement
$\mathcal{G}^{[1|1]}_{\mathcal{O}}\equiv \left\{\Pi^{[1|1]}_{\mathbf{a}}|~\Pi^{[1|1]}_{\mathbf{a}}\geq 0~ \&~ \sum_{\mathbf{a}}\Pi^{[1|1]}_{\mathbf{a}} = \mathbf{I}^{\otimes 2}_{2}\right\}$ on antiparallel configuration that jointly measures all observables in $N$ up to an optimal sharpness threshold $\lambda^{[1|1]}_{\mathcal{O}}$, i.e. $\forall~\rho_{\vec{m}}\in\mathcal{D}(\mathbb{C}^2)$ and $\forall~r$
\begin{align}
\sum_{\mathbf{a}\setminus a_{r}}\tr\left[\Pi^{[1|1]}_{\mathbf{a}}~ (\rho_{\vec{m}}\otimes \rho_{-\vec{m}})\right]&= \tr\left[\frac{1}{2}\left(\mathrm{P}_{a_r\hat{n}_r}(\lambda^{[1|1]}_{\mathcal{O}})\otimes\mathbf{I}_2+\mathbf{I}_2\otimes\mathrm{P}_{-a_r\hat{n}_r}(\lambda^{[1|1]}_{\mathcal{O}})\right)~ (\rho_{\vec{m}}\otimes\rho_{-\vec{m}})\right]\nonumber\\
&=\tr\left[ \mathrm{P}_{a_r\hat{n}_r}(\lambda^{[1|1]}_{\mathcal{O}})~ \rho_{\vec{m}} \right].\nonumber
\end{align}
Let $U_{\hat{n}}$ denotes the unitary operator that employs $\pi$ rotation of the Bloch sphere about the axis $\hat{n}$. For the unitary operator  $U_{\hat{y}}$ we have  
\begin{align}
&U_{\hat{y}}(n_0\mathbf{I}_2+n_x\sigma_x+n_y\sigma_y+n_z\sigma_z)U^\dagger_{\hat{y}}\nonumber\\
&=n_0\mathbf{I}_2-n_x\sigma_x+n_y\sigma_y-n_z\sigma_z,~~\forall~n_0,n_x,n_y,n_z\in\mathbb{R}.\nonumber \label{yuni} 
\end{align}
Consider now a set of operators $\{\tilde{\Pi}_{\mathbf{a}}\}_{\bf a}$, with $\tilde{\Pi}_{\mathbf{a}}:= \mathbf{I}_2\otimes U_{\hat{y}}\big(~\Pi^{[1|1]}_{\mathbf{a}}\big)\mathbf{I}_2\otimes U^\dagger_{\hat{y}}$. Manifestly, each of the operators $\tilde{\Pi}_{\mathbf{a}}$'s are positive and
\begin{align}
\sum_{\mathbf{a}}\tilde{\Pi}_{\mathbf{a}}&=\sum_{\mathbf{a}}\mathbf{I}_2\otimes U_{\hat{y}}\Big(~\Pi^{[1|1]}_{\mathbf{a}}\Big)\mathbf{I}_2\otimes U^\dagger_{\hat{y}}=\mathbf{I}_2\otimes U_{\hat{y}}\Big(\sum_{\mathbf{a}}\Pi^{[1|1]}_{\mathbf{a}}\Big)\mathbf{I}_2\otimes U^\dagger_{\hat{y}}=\mathbf{I}_2\otimes\mathbf{I}_2,\nonumber
\end{align}
thereby ensuring the set $\{\tilde{\Pi}_{\mathbf{a}}\}_{\bf a}$ to be a valid POVM. We will now show that this POVM when applied on two qubits prepared in parallel configuration reproduce the measurement statistics of all observables in $\mathcal{O}$ up to sharpness threshold $\lambda^{[1|1]}_{\mathcal{O}}$:    
\begin{align}
&\sum_{\mathbf{a}\setminus a_{r}}
\tr\left[\tilde{\Pi}_{\mathbf{a}}~(\rho_{\vec{m}}\otimes \rho_{\vec{m}})\right]\nonumber\\
&= \sum_{\mathbf{a}\setminus a_{r}}
\tr\left[\mathbf{I}_2\otimes U_{\hat{y}}\left(\Pi^{[1|1]}_{\mathbf{a}}\right)\mathbf{I}_2 \otimes U^\dagger_{\hat{y}}~ (\rho_{\vec{m}}\otimes \rho_{\vec{m}})\right]\nonumber\\
&= \tr\Big[\mathbf{I}_2 \otimes U_{\hat{y}}\Big(\sum_{\mathbf{a}\setminus a_{r}}\Pi^{[1|1]}_{\mathbf{a}}\Big)\mathbf{I}_2 \otimes U^\dagger_{\hat{y}}~ (\rho_{\vec{m}}\otimes\rho_{\vec{m}})\Big]\nonumber\\
&= \tr\left[\mathbf{I}_2\otimes U_{\hat{y}}\left(\tfrac{1}{2}\big(\mathrm{P}_{a_r\hat{n}_r}(\lambda^{[1|1]}_{\mathcal{O}})\otimes\mathbf{I}_2+\mathbf{I}_2\otimes\mathrm{P}_{-a_r\hat{n}_r}(\lambda^{[1|1]}_{\mathcal{O}})\big)\right)\mathbf{I}_2\otimes U^\dagger_{\hat{y}}~(\rho_{\vec{m}}\otimes\rho_{\vec{m}}) \right].\nonumber 
\end{align}
This further implies
\begin{align}
&\sum_{\mathbf{a}\setminus a_{r}}
\tr\left[\tilde{\Pi}_{\mathbf{a}}~(\rho \otimes \rho)\right]= \tr\left[\tfrac{1}{2}\big(\mathrm{P}_{a_r\hat{n}_r}(\lambda^{[1|1]}_{\mathcal{O}})\otimes\mathbf{I}_2+ \mathbf{I}_2\otimes\mathrm{P}_{a_r\hat{n}_r}(\lambda^{[1|1]}_{\mathcal{O}})\big)~(\rho_{\vec{m}}\otimes\rho_{\vec{m}}) \right]\nonumber\\
&\hspace{0.5cm}=\tr\left[\mathrm{P}_{a_r\hat{n}_r}(\lambda^{[1|1]}_{\mathcal{O}})~\rho_{\vec{m}}\right].
\end{align}
This establishes that the optimal sharpness threshold $\lambda^{[1|1]}_{\mathcal{O}}$ achieved on antiparallel configuration can be achieved on the parallel configuration as well, namely $\lambda^{[2]}_{\mathcal{O}}\ge\lambda^{[1|1]}_{\mathcal{O}}$.
To show that the converse is also true, we start by assuming a joint measurement
$\mathcal{G}^{[2]}_{\mathcal{O}}\equiv \left\{\Pi^{[2]}_{\mathbf{a}}|~\Pi^{[2]}_{\mathbf{a}}\geq 0~ \&~ \sum_{\mathbf{a}}\Pi^{[2]}_{\mathbf{a}} = \mathbf{I}^{\otimes 2}_{2}\right\}$ that on the parallel configuration jointly measures the observables in $\mathcal{O}$ up to an optimal sharpness threshold $\lambda^{[2]}_{\mathcal{O}}$, and then repeat the steps as above. This completes the proof.
\end{proof}

\section{Proof of the No-Comparison Theorem \ref{theo1}}
\begin{proof}
The observables in $\mathtt{SyTri}$ is trivially jointly measurable on three identical copies of a qubit state, implying $\lambda^{[3]}_{\mathtt{SyTri}}=1$. On the other hand, Proposition~\ref{prop2} implies $\lambda^{[1|1]}_{\mathtt{SyTri}}=\lambda^{[2]}_{\mathtt{SyTri}}$, which according to Proposition~\ref{prop1} is strictly less than unity. Proposition~\ref{prop1} also implies that $\lambda^{[3]}_{\mathtt{SyTet}}<1$. We are thus left with to show perfect joint measurability of the observables in $\mathtt{SyTet}$ on anti-parallel configuration, for which we will now provide an explicit construction. Without loss of any generality, we fix $\mathtt{SyTet}\equiv\{\sigma_{\hat{n}_r}~|~r=0,1,2,3\}$ with
\begin{align}
&\left\{\!\begin{aligned}
&\hat{n}_0=\tfrac{1}{\sqrt{3}}\big(1,1,1\big);~~~~~~~~~\hat{n}_1=\tfrac{1}{\sqrt{3}}\big(1,-1,-1\big);\\
&\hat{n}_2=\tfrac{1}{\sqrt{3}}\big(-1,1,-1\big);~~\hat{n}_3=\tfrac{1}{\sqrt{3}}\big(-1,-1,1\big)
\end{aligned}\right\},\nonumber
\end{align}
where $\delta_{ab}$ is the Kronecker delta function and bar denotes bit-wise not operation. Consider now a set of six rank-1 positive operators $\mathcal{G}^{[1|1]}_{\mathtt{SyTet}}\equiv\left\{\Pi^{[1|1]}_{\pm\hat{\alpha}}~|~\alpha=x,y,z\right\}$. All of this operators are positive and rank-1, as evident form the expression 
\begin{align}
\Pi^{[1|1]}_{\pm\hat{\alpha}}&=\tfrac{2}{3}\ket{\xi_{\pm\hat{\alpha}}}\bra{\xi_{\pm\hat{\alpha}}};\nonumber\\
\ket{\xi_{\hat{\alpha}}}&:=\tfrac{\sqrt{3}+1}{2\sqrt{2}}\ket{\hat{\alpha},-\hat{\alpha}}+\tfrac{\sqrt{3}-1}{2\sqrt{2}}\ket{-\hat{\alpha},\hat{\alpha}}.\nonumber    
\end{align}
Furthermore, expressing them in Hilbert Schmidt basis 
\begin{align}
\Pi^{[1|1]}_{\pm\hat{\alpha}}&=\tfrac{1}{12} \Big(2 \mathbf{I}_2^{\otimes 2} -2{\sigma}_{\hat{\alpha}}^{\otimes2}
\pm\sqrt{3}\dcb{\sigma_{\hat{\alpha}},\mathbf{I}_2}~+ {\sigma}_{\hat{\beta}}^{\otimes 2}+ {\sigma}_{\hat{\gamma}}^{\otimes 2}\Big),\nonumber
\end{align}
where $\alpha,\beta,\gamma\in\{x,y,z\}$ with $\alpha\neq\beta\neq\gamma$, and $\dcb{A,B}~:=A\otimes B-B\otimes A$, it immediately follows $\sum_{\alpha=x,y,z} \left(\Pi^{[1|1]}_{\hat{\alpha}}+\Pi^{[1|1]}_{-\hat{\alpha}}\right)=\mathbf{I}_2\otimes\mathbf{I}_2$, thereby ensuring $\mathcal{G}^{[1|1]}_{\mathtt{SyTet}}$ to be a valid POVM on $2$-qubit Hilbert space. As shown below, this POVM when applied on $\rho_{\hat{m}}\otimes\rho_{-\hat{m}}$ reproduces measurement statistics of $\mathtt{SyTet}$ observables on the state $\rho_{\hat{m}}$ 
\begin{align}
&\tr\left[\left(\Pi^{[1|1]}_{\hat{x}}+\Pi^{[1|1]}_{\hat{y}}+\Pi^{[1|1]}_{\hat{z}}\right)\left(\rho_{\hat{m}}\otimes\rho_{-\hat{m}}\right)\right]=\tr\left[\mathrm{P}_{\hat{n}_0}\rho_{\hat{m}}\right]\nonumber\\
&\tr\left[\left(\Pi^{[1|1]}_{\hat{x}}+\Pi^{[1|1]}_{-\hat{y}}+\Pi^{[1|1]}_{-\hat{z}}\right)\left(\rho_{\hat{m}}\otimes\rho_{-\hat{m}}\right)\right]=\tr\left[\mathrm{P}_{\hat{n}_1}\rho_{\hat{m}}\right]\nonumber\\
&\tr\left[\left(\Pi^{[1|1]}_{-\hat{x}}+\Pi^{[1|1]}_{\hat{y}}+\Pi^{[1|1]}_{-\hat{z}}\right)\left(\rho_{\hat{m}}\otimes\rho_{-\hat{m}}\right)\right]=\tr\left[\mathrm{P}_{\hat{n}_2}\rho_{\hat{m}}\right]\nonumber\\
&\tr\left[\left(\Pi^{[1|1]}_{-\hat{x}}+\Pi^{[1|1]}_{-\hat{y}}+\Pi^{[1|1]}_{\hat{z}}\right)\left(\rho_{\hat{m}}\otimes\rho_{-\hat{m}}\right)\right]=\tr\left[\mathrm{P}_{\hat{n}_3}\rho_{\hat{m}}\right]\nonumber   
\end{align}
This completes the proof.
\end{proof}

\noindent Here we would like to point out that this construction is motivated from the observation that the qubit POVM $\mathcal{G}^{[1]}_{\mathtt{SyTet}}$ consisting of the effects 
\begin{align}
\left\{\Pi^{[1]}_{\pm\hat{\alpha}}:=\tfrac{2}{3}\mathrm{P}_{\pm\hat{\alpha}}~|~\alpha\in\{x,y,z\}\right\}\nonumber
\end{align}
reproduces measurement statistics of the observables in $\mathtt{SyTet}$ up to an optimal sharpness threshold $\lambda=1/\sqrt{3}$, namely 
\begin{align}
\left\{\begin{aligned}
&\tfrac{2}{3}\mathrm{P}_{\hat{x}}+\tfrac{2}{3}\mathrm{P}_{\hat{y}}+\tfrac{2}{3}\mathrm{P}_{\hat{z}}=\mathrm{P}_{\hat{r}_0}(\lambda),\\
&\tfrac{2}{3}\mathrm{P}_{\hat{x}}+\tfrac{2}{3}\mathrm{P}_{-\hat{y}}+\tfrac{2}{3}\mathrm{P}_{-\hat{z}}=\mathrm{P}_{\hat{r}_1}(\lambda),\\
&\tfrac{2}{3}\mathrm{P}_{-\hat{x}}+\tfrac{2}{3}\mathrm{P}_{\hat{y}}+\tfrac{2}{3}\mathrm{P}_{-\hat{z}}=\mathrm{P}_{\hat{r}_2}(\lambda),\\
&\tfrac{2}{3}\mathrm{P}_{-\hat{x}}+\tfrac{2}{3}\mathrm{P}_{-\hat{y}}+\tfrac{2}{3}\mathrm{P}_{\hat{z}}=\mathrm{P}_{\hat{r}_3}(\lambda)
\end{aligned}\right\}.\nonumber 
\end{align}

\begin{remark}\label{remark1}
Notably, the Bloch vectors of the projectors in observable set $\mathrm{SyTet}\equiv \{\sigma_{\hat{n}_i}\}_{i=0}^3$ forms a symmetric cube embedded in the Bloch sphere, whereas Bloch vectors of the reduced part of the effects in $\mathcal{G}^{[1|1]}_{\mathtt{SyTet}}$ form a symmetric octahedron with its vertices sitting on the faces of this cube.    
\end{remark}
\begin{observation}\label{obs1}
Note that the generality of Propositions~\ref{prop1} and \ref{prop2} allow us to replace the set $\mathtt{SyTri}$ by any set of three distinct spin observables from an equatorial plane in the No-Comparison theorem, and still the same conclusion holds.   
\end{observation}

\subsection{Optimal Sharpness Thresholds: $\lambda^{[2]}_{\mathtt{SyTri}}$ and $\lambda^{[3]}_{\mathtt{SyTet}}$}
\subsubsection{Analysis of $\mathtt{SyTri}$}
\noindent As established in our Proposition~\ref{prop2} and also verified through SDPs (see Figure \ref{figs2}), $\lambda^{[2]}_{\mathtt{SyTri}}=\lambda^{[1|1]}_{\mathtt{SyTri}}$. Here we provide the optimal value of this sharpness parameter along with the POVM. We consider $\mathtt{SyTri}\equiv\mathcal{O}_{\pi/2}\equiv\{\sigma_{\hat{n}_r}~|~r=1,2,3\}$ with 
\begin{align}
\hat{n}_1= \frac{1}{\sqrt{6}}(2, -1, -1)^{\mathrm{T}},~
\hat{n}_2= \frac{1}{\sqrt{6}}(-1, 2, -1)^{\mathrm{T}},~\hat{n}_3= \frac{1}{\sqrt{6}}(-1, -1, 2)^{\mathrm{T}}.\nonumber
\end{align}

\noindent Effects of the optimal POVM $\mathcal{G}^{[2]}_{\mathtt{SyTri}}\equiv\left\{\Pi^{[2]}_{\bf a}~|~a_i=\pm1\right\}$ can be expressed as
\begin{align}
\Pi^{[2]}_{\bf a}&=\tfrac{1}{8}\left[\tfrac{1}{4}\left(\beta^{\mathbf{I}_2\mathbf{I}_2}_{\bf a}~\mathbf{I}^{\otimes2}_{2}\right)+\tfrac{2}{3\sqrt{3}}\left(\beta^{X\mathbf{I}_2}_{\bf a}~\dacb~{X,\mathbf{I}_2}+\beta^{Y\mathbf{I}_2}_{\bf a}~\dacb~{Y,\mathbf{I}_2}+\beta^{X\mathbf{I}_2}_{\bf a}~\dacb~{X,\mathbf{I}_2}\right)\right.\nonumber\\
&\hspace{0cm}\left.+\tfrac{2}{9}\left(\beta^{XY}_{\bf a}~\dacb~{X,Y}+\beta^{YZ}_{\bf a}~\dacb~{Y,Z}+\beta^{ZX}_{\bf a}~\dacb~{Z,X}\right)+\tfrac{1}{36}\left(\beta^{XX}_{\bf a}~X^{\otimes2}+\beta^{YY}_{\bf a}~Y^{\otimes2}+\beta^{ZZ}_{\bf a}~Z^{\otimes2}\right)\right],\nonumber
\end{align}
with $\dacb~{A,B}:=A\otimes B+B\otimes A$, and is completely specified by the coefficients provide in Table \ref{tab1}. A straightforward calculation ensures that the POVM $\mathcal{G}^{[2]}_{\mathtt{SyTri}}$ satisfies
\begin{align}
\sum_{\mathbf{a} \setminus a_r} \Pi^{[2]}_{\mathbf{a}} = \frac{1}{2} \left( \mathrm{P}_{a_r \hat{n}_r}\left(\tfrac{2\sqrt{2}}{3}\right) \otimes \mathbf{I}_2 + \mathbf{I}_2 \otimes \mathrm{P}_{a_r \hat{n}_r}\left(\tfrac{2\sqrt{2}}{3}\right)\right),\nonumber
\end{align}
for $r\in\{1,2,3\}$, thereby jointly measuring the triple spin observables $\mathcal{O}_{\pi/2}\equiv\{\sigma_{\hat{n}_r}\}$ up to the sharpness threshold $\lambda^{[2]}_{\mathtt{SyTri}}=\frac{2\sqrt{2}}{3}\approx0.94281$. 
\begin{table*}[h!]
\centering
\begin{tabular}{c||c|c|c|c|c|c|c|c|c|c|}
&~$\beta^{\mathbf{I}_2\mathbf{I}_2}_{\bf a}$~&~ $\beta^{X\mathbf{I}_2}_{\bf a}$~&~ $\beta^{Y\mathbf{I}_2}_{\bf a}$~ & ~$\beta^{Z\mathbf{I}_2}_{\bf a}$~&~$\beta^{XY}_{\bf a}$ ~&~ $\beta^{YX}_{\bf a}$ ~&~ $\beta^{ZX}_{\bf a}$ ~&~ $\beta^{XX}_{\bf a}$ ~&~$\beta^{YY}_{\bf a}$ ~&~ $\beta^{ZZ}_{\bf a}$ \\ \hline\hline

$\Pi^{[2]}_{+1+1+1}$ & $+1$ & $0$ & $0$ & $0$ & $0$ & $0$ & $0$ & $-9$ & $-9$ & $-9$  \\ \hline

$\Pi^{[2]}_{+1+1-1}$ & $+5$ & $+1$ & $+1$ & $-2$ & $+2$ & $-1$ & $-1$ & $-5$ & $-5$ & $+19$  \\ \hline

$\Pi^{[2]}_{+1-1+1}$ & $+5$ & $+1$ & $-2$ & $+1$ & $-1$ & $-1$ & $+2$ & $-5$ & $+19$ & $-5$  \\ \hline

$\Pi^{[2]}_{+1-1-1}$ & $+5$ & $+2$ & $-1$ & $-1$ & $-1$ & $+2$ & $-1$ & $+19$ & $-5$ & $-5$  \\ \hline

$\Pi^{[2]}_{-1+1+1}$ & $+5$ & $-2$ & $+1$ & $+1$ & $-1$ & $+2$ & $-1$ & $+19$ & $-5$ & $-5$  \\ \hline

$\Pi^{[2]}_{-1+1-1}$ & $+5$ & $-1$ & $+2$ & $-1$ & $-1$ & $-1$ & $+2$ & $-5$ & $+19$ & $-5$  \\ \hline

$\Pi^{[2]}_{-1-1+1}$ & $+5$ & $-1$ & $-1$ & $+2$ & $+2$ & $-1$ & $-1$ & $-5$ & $-5$ & $+19$  \\ \hline

$\Pi^{[2]}_{-1-1-1}$ & $+1$ & $0$ & $0$ & $0$ & $0$ & $0$ & $0$ & $-9$ & $-9$ & $-9$  \\ \hline
\end{tabular}
\caption{The coefficients of the effects are listed. It is immediate from this table to verify that $\sum_{\bf a}\Pi^{[2]}_{\bf a}=\mathbf{I}_2\otimes\mathbf{I}_2$.}\label{tab1}
\end{table*}

On the anti-parallel configuration we also have the optimal sharpness $\lambda^{[1|1]}_{\mathtt{SyTri}}=\frac{2\sqrt{2}}{3}$ with the optimal POVM $\mathcal{G}^{[1|1]}_{\mathtt{SyTri}}\equiv\left\{\Pi^{[1|1]}_{\bf a}~|~a_i=\pm1\right\}$, where the effects $\Pi^{[1|1]}_{\mathbf{a}}$'s are obtained from the optimal parallel effects by the rule
\begin{align}
\Pi^{[1|1]}_{\mathbf{a}}:=\left(\mathbf{I}_2\otimes U_{\hat{t}}\right)\Pi^{[2]}_{\mathbf{a}}\left(\mathbf{I}_2\otimes U_{\hat{t}}\right),\nonumber
\end{align}
where $U_{\hat{t}}$ is the unitary operator employing $\pi$ rotation of the Bloch sphere about the axis $\hat{t}=\tfrac{1}{\sqrt{3}}(1,1,1)$.

\subsection{Analysis of $\mathtt{SyTet}$}

\noindent The set $\mathtt{SyTet}$ contains four spin observables $\{\sigma_{\hat{n}_r}~|~r=0,1,2,3\}$, where
\begin{align*}
\hat{n}_0=\tfrac{1}{\sqrt{3}}\big(1,1,1\big),~~~\hat{n}_1=\tfrac{1}{\sqrt{3}}\big(1,-1,-1\big),~~~\hat{n}_2=\tfrac{1}{\sqrt{3}}\big(-1,1,-1\big),~~~\hat{n}_3=\tfrac{1}{\sqrt{3}}\big(-1,-1,1\big).
\end{align*}
Here we use the following short hand notations 
\begin{align*}
\mathbb{S}(A,B,C)&:=A\otimes B\otimes C + B\otimes C\otimes A + C\otimes A\otimes B \nonumber\\
&\hspace{.2cm}+A\otimes C\otimes B + C\otimes B\otimes A + B\otimes A\otimes C,\\
\mathbb{S}(A,B,B)&:=A\otimes B\otimes B + B\otimes A\otimes B + B\otimes B\otimes A.
\end{align*}    
The optimal joint measurability threshold $\lambda^{[3]}_{\mathrm{SyTet}}$ is obtained through SDP. Expressing $\Pi^{[3]}_{\bf a}$'s in Hilbert–Schmidt operator basis as
{\footnotesize
\begin{align}
\Pi^{[3]}_{\bf a}=&\beta_{\bf a}^{X\mathbf{I}_2\mathbf{I}_2}~\mathbb{S}(X,\mathbf{I}_2,\mathbf{I}_2)+\beta_{\bf a}^{Y\mathbf{I}_2\mathbf{I}_2}~\mathbb{S}(Y,\mathbf{I}_2,\mathbf{I}_2)+\beta_{\bf a}^{Z\mathbf{I}_2\mathbf{I}_2}~\mathbb{S}(Z,\mathbf{I}_2,\mathbf{I}_2)+\beta_{\bf a}^{XX\mathbf{I}_2}~\mathbb{S}(\mathbf{I}_2,X,X)+\beta_{\bf a}^{YY\mathbf{I}_2}~\mathbb{S}(\mathbf{I}_2,Y,Y)+\beta_{\bf a}^{ZZ\mathbf{I}_2}~\mathbb{S}(\mathbf{I}_2,Z,Z)\nonumber\\
&+\beta_{\bf a}^{XY\mathbf{I}_2}~\mathbb{S}(X,Y,\mathbf{I}_2)+\beta_{\bf a}^{YZ\mathbf{I}_2}~\mathbb{S}(Y,Z,\mathbf{I}_2)+\beta_{\bf a}^{ZX\mathbf{I}_2}~\mathbb{S}(Z,X,\mathbf{I}_2)+\beta_{\bf a}^{XYY}~\mathbb{S}(X,Y,Y)+\beta_{\bf a}^{YXX}~\mathbb{S}(Y,X,X)+\beta_{\bf a}^{YZZ}~\mathbb{S}(Y,Z,Z)\nonumber\\
&+\beta_{\bf a}^{ZYY}~\mathbb{S}(Z,Y,Y)+\beta_{\bf a}^{ZXX}~\mathbb{S}(Z,X,X)+\beta_{\bf a}^{XZZ}~\mathbb{S}(X,Z,Z)+\beta_{\bf a}^{XXX}~X^{\otimes3}+\beta_{\bf a}^{YYY}~Y^{\otimes3}+\beta_{\bf a}^{ZZZ}~Z^{\otimes3}+\beta_{\bf a}^{\mathbf{I}_2\mathbf{I}_2\mathbf{I}_2}~\mathbf{I}_2^{\otimes3},\nonumber
\end{align}}
\nonumber the optimal POVM $\mathcal{G}^{[3]}_{\mathtt{SyTet}}\equiv\{\Pi^{[3]}_{\bf a}=\Pi^{[3]}_{a_0a_1a_2a_3}~|~a_i=\pm1\}$ is completely specified by coefficients listed in Table \ref{tab2}. A tedious but straightforward calculation shows that the POVM $\mathcal{G}^{[3]}_{\mathtt{SyTet}}$ satisfies
\begin{align}
\sum_{\mathbf{a} \setminus a_r} \Pi^{[3]}_{\mathbf{a}}=\frac{1}{6}\left[3~\mathbf{I}_2^{\otimes3} +\left(\frac{3}{\sqrt{3}+\sqrt{2}}\right)\mathbb{S}(\sigma_{\hat{n}_r},\mathbf{I}_2,\mathbf{I}_2)\right],
\end{align}
\noindent for $r\in\{0,1,2,3\}$, thereby ensuring joint measurability of the observables in $\mathtt{SyTet}$ on three copy ensemble $[3]$ up to the sharpness threshold $\lambda^{[3]}_{\mathrm{SyTet}}=\frac{3}{\sqrt{3}+\sqrt{2}}\approx0.95351$~.

\begin{table*}[t!]
\centering
\begin{tabular}{c||c|c|c|c|c|c|c|c|c|c|c|c|c|c|c|c|c|c|c|}
 &{\rotatebox{90}{$\beta_{\bf a}^{\mathbf{I}_2\mathbf{I}_2\mathbf{I}_2}$}} & {\rotatebox{90}{$\beta_{\bf a}^{X\mathbf{I}_2\mathbf{I}_2}$}} & {\rotatebox{90}{$\beta_{\bf a}^{Y\mathbf{I}_2\mathbf{I}_2}$}} & {\rotatebox{90}{$\beta_{\bf a}^{Z\mathbf{I}_2\mathbf{I}_2}$}} & {\rotatebox{90}{$\beta_{\bf a}^{XX\mathbf{I}_2}$}} & {\rotatebox{90}{$\beta_{\bf a}^{YY\mathbf{I}_2}$}} & {\rotatebox{90}{$\beta_{\bf a}^{ZZ\mathbf{I}_2}$}} & {\rotatebox{90}{$\beta_{\bf a}^{XY\mathbf{I}_2}$}} & {\rotatebox{90}{$\beta_{\bf a}^{YZ\mathbf{I}_2}$}} & {\rotatebox{90}{$\beta_{\bf a}^{ZX\mathbf{I}_2}$}} & {\rotatebox{90}{$\beta_{\bf a}^{XYY}$}} & {\rotatebox{90}{$\beta_{\bf a}^{YXX}$}} & {\rotatebox{90}{$\beta_{\bf a}^{YZZ}$}} & {\rotatebox{90}{$\beta_{\bf a}^{ZYY}$}} & {\rotatebox{90}{$\beta_{\bf a}^{ZXX}$}} & {\rotatebox{90}{$\beta_{\bf a}^{XZZ}$}} & {\rotatebox{90}{$\beta_{\bf a}^{XXX}$}} & {\rotatebox{90}{$\beta_{\bf a}^{Y~YY}$}}& {\rotatebox{90}{$\beta_{\bf a}^{ZZZ}$}} \\ \hline\hline

{$\Pi^{[3]}_{+1+1+1+1}$} & {$3a$} & {$0$} & {$0$} & {$0$} & {$-a$} & {$-a$} & {$-a$} & {$0$} & {$0$} & {$0$} & {$0$} & {$0$} & {$0$} & {$0$} & {$0$} & {$0$} & {$0$} & {$0$} & {$0$}  \\ \hline

{$\Pi^{[3]}_{-1-1-1-1}$} & {$3a$} & {$0$} & {$0$} & {$0$} & {$-a$} & {$-a$} & {$-a$} & {$0$} & {$0$} & {$0$} & {$0$} & {$0$} & {$0$} & {$0$} & {$0$} & {$0$} & {$0$} & {$0$} & {$0$}  \\ \hline

{$\Pi^{[3]}_{+1-1-1-1}$} & {$b_1$} & {$b_2$} & {$b_2$} & {$b_2$} & {$-b_3$} & {$-b_3$} & {$-b_3$} & {$b_4$} & {$b_4$} & {$b_4$} & {$b_5$} & {$b_5$} & {$b_5$} & {$b_5$} & {$b_5$} & {$b_5$} & {$-b_6$} & {$-b_6$} & {$-b_6$} \\ \hline

{$\Pi^{[3]}_{-1+1+1+1}$} & {$b_1$} & {$-b_2$} & {$-b_2$} & {$-b_2$} & {$-b_3$} & {$-b_3$} & {$-b_3$} & {$b_4$} & {$b_4$} & {$b_4$} & {$-b_5$} & {$-b_5$} & {$-b_5$} & {$-b_5$} & {$-b_5$} & {$-b_5$} & {$b_6$} & {$b_6$} & {$b_6$} \\ \hline

{$\Pi^{[3]}_{-1+1-1-1}$} & {$b_1$} & {$b_2$} & {$-b_2$} & {$-b_2$} & {$-b_3$} & {$-b_3$} & {$-b_3$} & {$-b_4$} & {$b_4$} & {$-b_4$} & {$b_5$} & {$-b_5$} & {$-b_5$} & {$-b_5$} & {$-b_5$} & {$b_5$} & {$-b_6$} & {$b_6$} & {$b_6$} \\ \hline

{$\Pi^{[3]}_{+1-1+1+1}$} & {$b_1$} & {$-b_2$} & {$b_2$} & {$b_2$} & {$-b_3$} & {$-b_3$} & {$-b_3$} & {$-b_4$} & {$b_4$} & {$-b_4$} & {$-b_5$} & {$b_5$} & {$b_5$} & {$b_5$} & {$b_5$} & {$-b_5$} & {$b_6$} & {$-b_6$} & {$-b_6$} \\ \hline

{$\Pi^{[3]}_{-1-1+1-1}$} & {$b_1$} & {$-b_2$} & {$b_2$} & {$-b_2$} & {$-b_3$} & {$-b_3$} & {$-b_3$} & {$-b_4$} & {$-b_4$} & {$b_4$} & {$-b_5$} & {$b_5$} & {$b_5$} & {$-b_5$} & {$-b_5$} & {$-b_5$} & {$b_6$} & {$-b_6$} & {$b_6$} \\ \hline

{$\Pi^{[3]}_{+1+1-1+1}$} & {$b_1$} & {$b_2$} & {$-b_2$} & {$b_2$} & {$-b_3$} & {$-b_3$} & {$-b_3$} & {$-b_4$} & {$-b_4$} & {$b_4$} & {$b_5$} & {$-b_5$} & {$-b_5$} & {$b_5$} & {$b_5$} & {$b_5$} & {$-b_6$} & {$b_6$} & {$-b_6$} \\ \hline

{$\Pi^{[3]}_{-1-1-1+1}$} & {$b_1$} & {$-b_2$} & {$-b_2$} & {$b_2$} & {$-b_3$} & {$-b_3$} & {$-b_3$} & {$b_4$} & {$-b_4$} & {$-b_4$} & {$-b_5$} & {$-b_5$} & {$-b_5$} & {$b_5$} & {$b_5$} & {$-b_5$} & {$b_6$} & {$b_6$} & {$-b_6$} \\ \hline

{$\Pi^{[3]}_{+1+1+1-1}$} & {$b_1$} & {$b_2$} & {$b_2$} & {$-b_2$} & {$-b_3$} & {$-b_3$} & {$-b_3$} & {$b_4$} & {$-b_4$} & {$-b_4$} & {$b_5$} & {$b_5$} & {$b_5$} & {$-b_5$} & {$-b_5$} & {$b_5$} & {$-b_6$} & {$-b_6$} & {$b_6$} \\ \hline

{$\Pi^{[3]}_{-1-1+1+1}$} & {$c_1$} & {$-c_2$} & {$0$} & {$0$} & {$c_3$} & {$-c_4$} & {$-c_4$} & {$0$} & {$0$} & {$0$} & {$c_5$} & {$0$} & {$0$} & {$0$} & {$0$} & {$c_5$} & {$-c_6$} & {$0$} & {$0$} \\ \hline

{$\Pi^{[3]}_{+1+1-1-1}$} & {$c_1$} & {$c_2$} & {$0$} & {$0$} & {$c_3$} & {$-c_4$} & {$-c_4$} & {$0$} & {$0$} & {$0$} & {$-c_5$} & {$0$} & {$0$} & {$0$} & {$0$} & {$-c_5$} & {$c_6$} & {$0$} & {$0$} \\ \hline

{$\Pi^{[3]}_{-1+1-1+1}$} & {$c_1$} & {$0$} & {$-c_2$} & {$0$} & {$-c_4$} & {$c_3$} & {$-c_4$} & {$0$} & {$0$} & {$0$} & {$0$} & {$c_5$} & {$c_5$} & {$0$} & {$0$} & {$0$} & {$0$} & {$-c_6$} & {$0$} \\ \hline

{$\Pi^{[3]}_{+1-1+1-1}$} & {$c_1$} & {$0$} & {$c_2$} & {$0$} & {$-c_4$} & {$c_3$} & {$-c_4$} & {$0$} & {$0$} & {$0$} & {$0$} & {$-c_5$} & {$-c_5$} & {$0$} & {$0$} & {$0$} & {$0$} & {$c_6$} & {$0$} \\ \hline

{$\Pi^{[3]}_{-1+1+1-1}$} & {$c_1$} & {$0$} & {$0$} & {$-c_2$} & {$-c_4$} & {$-c_4$} & {$c_3$} & {$0$} & {$0$} & {$0$} & {$0$} & {$0$} & {$0$} & {$c_5$} & {$c_5$} & {$0$} & {$0$} & {$0$} & {$-c_6$} \\ \hline

{$\Pi^{[3]}_{+1-1-1+1}$} & {$c_1$} & {$0$} & {$0$} & {$c_2$} & {$-c_4$} & {$-c_4$} & {$c_3$} & {$0$} & {$0$} & {$0$} & {$0$} & {$0$} & {$0$} & {$-c_5$} & {$-c_5$} & {$0$} & {$0$} & {$0$} & {$c_6$} \\ \hline
 
\end{tabular}
\caption{Coefficients of the operators $\Pi^{[3]}_{\bf a}$'s. The coefficient $\beta_{\bf a}^{XYZ}$ for all the $\Pi^{[3]}_{\bf a}$'s take value zero, and hence are not shown as a different column. Values of these coefficients are obtained through SDP and given by: $a=6k\big(2 \sqrt{2}-2 \sqrt{3}+1\big),$~$b_1=9k\big(10 \sqrt{2}-14 \sqrt{3}-4 \sqrt{6}+25\big),$~$
c_1=6k \big(-22 \sqrt{2}+30 \sqrt{3}+8 \sqrt{6}-35\big),$~$
b_2=4k \big(8 \sqrt{3}-4 \sqrt{3 \big(\sqrt{3}+2\big)}+3\big),$~$ 
c_2=8k \big(6 \sqrt{2}-8 \sqrt{3}-10 \sqrt{6}+33\big),$~$
b_3=3k \big(10 \sqrt{2}-2 \sqrt{3} \big(6 \sqrt{2}+7\big)+41\big),$~$
c_3=2k \big(22 \sqrt{2}-30 \sqrt{3}-40 \sqrt{6}+131\big),$~$
b_4=12k \big(\sqrt{6}-2\big),$~$
c_4=2k \big(-22 \sqrt{2}+30 \sqrt{3}+16 \sqrt{6}-59\big),$~$
b_5=4k \big(-8 \sqrt{3}-2 \sqrt{6 \big(7-4 \sqrt{3}\big)}+15\big),$~$
c_5=8k \big(-8 \sqrt{3}-2 \sqrt{6 \big(7-4 \sqrt{3}\big)}+15\big),$~$
b_6=12k \big(-6 \sqrt{2}+8 \sqrt{3}+2 \sqrt{6}-9\big),$~$
c_6=24k \big(-6 \sqrt{2}+8 \sqrt{3}+2 \sqrt{6}-9\big)$,~$k=1/576$~.}\label{tab2}
\end{table*}

\end{document}